\documentclass[12pt,draft,onecolumn]{IEEEtran}
\usepackage{amsfonts}
\usepackage{amsmath}
\usepackage{amssymb}
\usepackage{algorithm}
\usepackage{booktabs}
\usepackage{graphicx,color}
\usepackage[normalem]{ulem}

\newtheorem{theorem}{Theorem}

\newtheorem{remark}{Remark}
\newtheorem{assumption}{Assumption}
\newtheorem{lemma}{Lemma}
\newtheorem{definition}{Definition}

\newcommand{\ev}{{V}}
\newcommand{\sol}{{sol}}
\newcommand{\ic}{{k}} 
\newcommand{\ms}{{\textbf{q}}}
\newcommand{\ped}[1]{{_{\mathrm{#1}}}}

 \definecolor{darkgreen}{rgb}{0.0, 0.4, 0.1}

\IEEEoverridecommandlockouts

\title{\LARGE \bf
Sequential Randomized Algorithms for Convex Optimization in the Presence of Uncertainty
}
\author{Mohammadreza Chamanbaz, Fabrizio Dabbene, Roberto Tempo,\\
 Venkatakrishnan Venkataramanan, Qing-Guo Wang  
\thanks{M.\ Chamanbaz and V.\ Venkatakrishnan are with Data Storage Institute, Singapore (emails: Mrchamanbaz@gmail.com, Venka\_V@dsi.a-star.edu.sg).
        }%
\thanks{F.\ Dabbene and R.\ Tempo are with CNR-IEIIT Torino, Italy (emails: \{roberto.tempo, fabrizio.dabbene\}@polito.it).}%
\thanks{M.\ Chamanbaz and Q-G.\ Wang are with Department of Electrical and Computer Engineering, National University of Singapore (emails: \{Chamanbaz,Elewqg\}@nus.edu.sg). \newline
The results presented in this paper were obtained when the first author was visiting CNR-IEIIT for a period of six months under DSI funding.}}

\begin{document}

\maketitle

\begin{abstract}
In this paper, we propose new sequential randomized algorithms
for convex optimization problems in the presence of uncertainty.
A rigorous analysis of the theoretical properties of the solutions obtained by these algorithms, for full constraint satisfaction and partial constraint satisfaction, respectively, is given.
The proposed methods allow to enlarge the applicability of the existing randomized methods to real-world applications involving a large number of design variables.
Since the proposed approach does not provide a priori bounds on the sample complexity, extensive numerical simulations,
dealing with an application to hard-disk drive servo design, are provided. These simulations testify the goodness of the proposed solution.
\end{abstract}
\section{Introduction}\label{sec:introduction}
In recent years, research on randomized and probabilistic methods for control of uncertain systems has successfully evolved along various directions, see e.g. \cite{tempo_randomized_2012}
for an overview of the state of the art on this topic.
For convex control design, two main classes of algorithms, sequential and non-sequential,  have been proposed in the literature, and their theoretical properties have been rigorously studied, see e.g. \cite{calafiore_research_2011}.

Regarding non-sequential methods, the approach that has emerged is the so-called scenario approach, which has been introduced in
\cite{calafiore_uncertain_2004,calafiore_scenario_2006}.  Taking random samples of the uncertainty $q \in \mathbb{Q}$, the main idea of this particular line of research is to reformulate a semi-infinite convex optimization problem as a sampled
optimization problem subject to a finite number of random constraints.
Then, a key problem is to determine the sample complexity, i.e. the number of random constraints that should be generated, so that the so-called probability of violation is smaller than a given accuracy
$\varepsilon \in (0,1)$, and this event holds with a suitably large confidence $1-\delta \in (0,1)$.
On the other hand,
if accuracy and confidence are very small, and the number of design parameters is large, then the sample complexity may be large, and the sampled convex optimization problem may be difficult to solve in practice.

Motivated by this discussion, in this paper we develop a sequential method specifically tailored to the solution of the scenario-based optimization problem. The proposed approach iteratively solves \textit{reduced-size} scenario problems of increasing size, and it is particularly appealing for large-size problems.
This line of research follows and improves upon the schemes
previously developed for various control problems,  which include linear quadratic regulators, linear matrix inequalities and switched systems discussed in  \cite{calafiore_research_2011,tempo_randomized_2012}. The main idea of  these sequential methods is to introduce the concept of validation samples. That is, at step $k$ of the sequential algorithm, a ``temporary solution" is constructed and, using a suitably generated validation sample set, it is verified whether or not the probability of violation corresponding to the temporary solution is smaller than a given accuracy $\varepsilon$, and this event holds with confidence $1-\delta$.
Due to their sequential nature, these algorithms may have wider practical applications than non-sequential methods, in particular for real-world problems where fast computations are needed because of very  stringent time requirements due to on-line implementations.
However, we remark that the sequential methods proposed here, contrary to the scenario approach previously discussed,
do not provide a priori bounds on the sample complexity.

Compared to the sequential approaches discussed above, the methods proposed in this paper have the following distinct main advantages:
1. the termination of the algorithm does not require the knowledge of some user-determined parameters, such as the center of a feasibility ball; 2. the methods can be immediately implemented using existing
off-the-shelf convex optimization tools, and no ad-hoc implementation of specific update rules (such as stochastic gradient, ellipsoid or cutting plane) is needed.
We also remark that the methods presented here directly apply to optimization problems, whereas the sequential methods discussed in  \cite{calafiore_research_2011,tempo_randomized_2012} are limited to feasibility.

In this paper, which is an expanded version of \cite{chamanbaz_sequential_2013}, we study two new sequential algorithms for optimization,
with  full constraint satisfaction and partial constraint satisfaction, respectively, and we provide a rigorous analysis of their theoretical properties
regarding the probability of violation of the returned solutions. These algorithms fall into the class of sequential probabilistic validation (SPV) algorithms introduced in \cite{alamo_randomized_2015}.


In the second part of the paper, using  a non-trivial example regarding the position control of read/write head in a commercial hard disk drive, we provide  extensive numerical simulations to compare the
sample complexity of the scenario approach with the number of iterations required in the two sequential algorithms previously introduced. We remark that the sample complexity of the scenario approach is computed a priori, while for sequential algorithms, the numerical results  regarding the size of the validation sample set are random. For this reason, mean values, standard deviation and other related parameters are experimentally computed for both proposed algorithms by means of extensive Monte Carlo simulations. We also highlight that the worst case complexity of the proposed methods may be larger than that of the scenario approach.

\section{Problem Formulation and Preliminaries}\label{sec:formulation}
An uncertain convex problem has the form
\begin{align}\label{eq:original problem}
& \underset{\theta\in\Theta}{\text{min}}\quad  c^T\theta\\ \nonumber
& \text{subject to }   f(\theta,q)\leq0\text{ for all }q\in\mathbb{Q}
\label{eq:problem}
\end{align}
\noindent
where $\theta\in\Theta\subset\mathbb{R}^{n_\theta}$ is the vector of optimization variables and $q\in\mathbb{Q}$ denotes random uncertainty acting on the system, $f(\theta,q):\Theta\times\mathbb{Q}\to\mathbb{R}$ is convex in $\theta$ for any fixed value of $q\in\mathbb{Q}$ and $\Theta$ is a convex and closed set. We note that most uncertain convex problems can be reformulated as (\ref{eq:original problem}). In particular, multiple scalar-valued constraints $f_i(\theta,q)\leq0,\,\,i=1,\,\ldots,\,m$ can always be recast into the form (\ref{eq:original problem}) by defining $f(\theta,q)=\underset{i=1,\,\ldots,\,m}{\max}\,f_i(\theta,q)$.

In this paper, we study a probabilistic framework where the uncertainty vector $q$ is assumed to be a random variable and the constraint in (\ref{eq:original problem}) is allowed to be violated for some $q\in\mathbb{Q}$, provided that the rate of violation is sufficiently small. This concept is formally expressed  using the notion of ``probability of violation".
\begin{definition}[Probability of Violation]
The probability of violation of $\theta$ for the function $f:\Theta\times\mathbb{Q}\to\mathbb{R}$ is defined as
\begin{equation}\label{eq:prob of viol}
\ev(\theta)\doteq \Pr\left\{q\in\mathbb{Q}\,\,:\,f(\theta,q)>0\right\}.
\end{equation}
\end{definition}
The exact computation of $\ev(\theta)$ is in general very difficult since it requires the computation of multiple integrals associated to the probability in (\ref{eq:prob of viol}). However,  this probability can be estimated using randomization. To this end, assuming that a probability measure is given over the set $\mathbb{Q}$, we generate $N$ independent identically distributed (i.i.d.) samples within the set $\mathbb{Q}$
\[
\ms=\{q^{(1)},\ldots,q^{(N)}\}\in\mathbb{Q}^N,
\]
where $\mathbb{Q}^{N}\doteq \mathbb{Q}\times\mathbb{Q}\times\cdots\times\mathbb{Q}$ ($N$ times). Next, a Monte Carlo approach is employed to obtain the so called ``empirical violation" which is introduced in the following definition.
\begin{definition}[Empirical Violation]
For given $\theta\in\Theta$ the empirical violation of $f(\theta,q)$ with respect to the multisample $\ms=\{q^{(1)},\ldots,q^{(N)}\}$ is defined as
\begin{equation}\label{eq:emp mean}
\widehat{V}(\theta,\ms)\doteq\frac{1}{N}\sum _{i=1}^N \mathbb{I}_f(\theta,q^{(i)})
\end{equation}
where $\mathbb{I}_f(\theta,q^{(i)})$ is an indicator function defined as 
\[
\mathbb{I}_f(\theta,q^{(i)})\doteq\begin{cases}
0\quad \text{if}\,f(\theta,q^{(i)})\leq0\\
1\quad \text{otherwise}
\end{cases}.
\]
\end{definition}
\noindent
\subsection{The Scenario Approach}\label{sec:scenario}
In this subsection, we briefly recall the so-called scenario approach, also known as random convex programs, which was first introduced in \cite{calafiore_uncertain_2004,calafiore_scenario_2006}, see also \cite{campi_exact_2008} for additional results.
In this approach, a set of independent identically distributed random samples of cardinality $N$ is extracted from the uncertainty set and the following scenario problem is formed
\begin{align}\label{eq:scenario design}
& \underset{\theta\in\Theta}{\text{min}}\quad c^T\theta\\\nonumber
& \text{subject to }\, f(\theta,q^{(i)})\leq0,\,i=1,\ldots,N.
\end{align}
The function $f(\theta,q)$ is convex for fixed $q\in\mathbb{Q}$ and a further assumption is that the problem (\ref{eq:scenario design}) is feasible for any finite number of samples and attains a unique solution $\widehat{\theta}_N$. These assumptions are now formally stated.
\begin{assumption}[Convexity]\label{assump:convexity}
$\Theta\subset\mathbb{R}^{n_\theta}$ is a convex and closed set and $f(\theta,q)$ is convex in $\theta$ for any fixed value of $q\in\mathbb{Q}$.
\end{assumption}
\begin{assumption}[Feasibility and Uniqueness]\label{assump:unique}
The sampled optimization problem (\ref{eq:scenario design}) is feasible for any multisample extraction and its feasibility domain has a nonempty interior. Furthermore, the solution of (\ref{eq:scenario design}) exists and is unique.
\end{assumption}

We remark that the uniqueness assumption can be relaxed in most cases by introducing a tie-breaking rule (see Section 4.1 of \cite{calafiore_uncertain_2004}).
The probabilistic property of the optimal solution obtained from (\ref{eq:scenario design}) is stated in the next lemma taken from \cite{campi_exact_2008}.

\begin{lemma}\label{lemma:scenario binomial}
Let Assumptions \ref{assump:convexity} and \ref{assump:unique} hold and let $\delta,\,\varepsilon\in(0,1)$ and $N$ satisfy the following inequality
\begin{equation}\label{eq:binomial dis}
\sum_{i=0}^{n_\theta-1} { N \choose i}
\varepsilon^i(1-\varepsilon)^{N-i}\leq\delta.
\end{equation}
Then, with probability at least $1-\delta$,  the solution of the optimization problem (\ref{eq:scenario design})  $\widehat{\theta}_N$  satisfies the inequality $\ev(\widehat{\theta}_N)\leq\varepsilon$.
\end{lemma}
We remark that Assumption 2, which guarantees that the sample problem is feasible,
is rather common in the literature on random convex programs, and can be relaxed using the approach introduced in \cite{calafiore_random_2010}.
In particular, in this case $n_{\theta}-1$ in (\ref{eq:binomial dis}) should be replaced by $n_{\theta}$.

\subsection{Scenario with Discarded Constraints}\label{sec:scenario with discarded}
The idea of scenario with discarded constraints \cite{calafiore_random_2010,CamGar:11} is to generate $N$ i.i.d. samples and then purposely discard $r<N-n_\theta$ of them. In other words, we solve an optimization problem of the form
\begin{align}\label{eq:random convex discarded}
& \underset{\theta\in\Theta}{\text{min}}\quad c^T\theta\\\nonumber
& \text{subject to }\, f(\theta,q^{(i)})\leq0,\,i=1,\ldots,N-r,
\end{align}
where, for notation ease, we assumed that the discarded constraints correspond to the last $r$ ones\footnote{In the more general case, the constraint in \eqref{eq:random convex discarded} should be written as follows
$f(\theta,q^{(i_v)})\leq0,\,v=1,\ldots,N-r,$
where $i_v\in\{1,\ldots,N\}$ represent the not discarded constraints. Note that this assumption is made without loss of generality, since the two sets of constraints are equivalent up to a  reordering.}.

The $r$ discarded samples are chosen so that the largest improvement in the optimal objective value is achieved. We remark that the optimal strategy to select $r$ discarded samples is a mixed-integer optimization problem, which may be hard to solve numerically. The following lemma \cite{CamGar:11} defines the probabilistic properties of the optimal solution obtained from (\ref{eq:random convex discarded}).
\begin{lemma}\label{lemma:discarded scenario}
Let Assumptions \ref{assump:convexity} and \ref{assump:unique} hold and let $\delta,\,\varepsilon\in(0,1)$, $N$  and $r<N-n_\theta$ satisfy the following inequality
\begin{equation}\label{eq:binomial discarded constraints}
{r+n_\theta-1 \choose r}\sum_{i=0}^{r+n_\theta-1} { N \choose i}
\varepsilon^i(1-\varepsilon)^{N-i}\leq\delta.
\end{equation}
Then, with probability at least $1-\delta$, the optimal solution of the optimization problem (\ref{eq:random convex discarded})  $\widehat{\theta}_N$ satisfies the inequality $\ev(\widehat{\theta}_N)\leq\varepsilon$.
\end{lemma}

Note that in the literature there are different results regarding explicit sample complexity bounds  $N$ such that (\ref{eq:binomial dis})
or (\ref{eq:binomial discarded constraints}) are satisfied for given values of $\varepsilon,\delta\in(0,1)$, see e.g.\ \cite{alamo_randomized_2015,alamo_sample_2010,  calafiore_random_2010}.
These  bounds
depend linearly on $1/\varepsilon$ and $n_{\theta}$ and logarithmically on $1/\delta$. However, in practice, the required number of samples can be very large even for problems with moderate number of decision variables. Therefore, the computational load  of the random convex problems (\ref{eq:scenario design}) and (\ref{eq:random convex discarded}) might be beyond the capability of the available convex optimization solvers. Motivated by this observation, in the next section we propose two  sequential randomized algorithms for optimization.

\section{Sequential Randomized Algorithms}\label{sec:sequential algorithm}
The main philosophy behind the proposed sequential randomized algorithms lies on the fact that it is easy from the  computational point of view to evaluate a given ``candidate solution" for a large number of random samples extracted from $\mathbb{Q}$. On the other hand, it is clearly more expensive to solve the optimization problems (\ref{eq:scenario design}) or (\ref{eq:random convex discarded}) when the sample bound $N$ is large. The sequential randomized algorithms, which are presented next generate a sequence of ``design" sample sets $\{q_d^{(1)},\ldots,q_d^{(N_\ic)}\}$ with increasing cardinality $N_\ic$ which are used in (\ref{eq:scenario design}) and (\ref{eq:random convex discarded}) for solving the optimization problem. In parallel, ``validation" sample sets  $\{q_v^{(1)},\ldots,q_v^{(M_\ic)}\}$ of cardinality $M_\ic$ are also generated by both algorithms in order to check whether the given candidate solution, obtained from solving (\ref{eq:scenario design}) or~(\ref{eq:random convex discarded}), satisfies the desired violation probability.

The first algorithm is in line with those presented in \cite{calafiore_probabilistic_2007} and \cite{yasuaki_polynomial-time_2007},  in the sense that it uses a similar strategy to validate the candidate solution. However, while these algorithms have been designed for feasibility problems, the proposed algorithms deal with optimization problems.

\subsection{Full Constraint Satisfaction}\label{sec:full constrant}
The first sequential randomized algorithm is presented in Algorithm~\ref{alg:sequential algorithm 1}, and its theoretical properties are  stated in the  following theorem.
\begin{algorithm}�
\caption{\sc{Sequential Randomized Algorithm: Full Constraint Satisfaction}}
\label{alg:sequential algorithm 1}
\begin{enumerate}
  \item \textsc{Initialization} \newline
  Set  iteration counter to zero $(\ic=0)$. Choose  probabilistic levels $\varepsilon$, $\delta$ and number of iterations $\ic_t>1$.

  \item \textsc{Update}\label{item:update 1}\newline
  Set $\ic=\ic+1$ and $N_\ic\ge N\frac{\ic}{\ic_t}$ where $N$ is the smallest integer satisfying
  \begin{equation}\label{eq:binomial dis2}
\sum_{i=0}^{n_\theta-1} { N \choose i}
\varepsilon^i(1-\varepsilon)^{N-i}\leq\delta/2.
\end{equation}
  \item \textsc{Design}
  \begin{itemize}
    \item Draw $N_\ic$ i.i.d. samples $\ms_d=\{q_d^{(1)},\ldots, q_d^{(N_\ic)}\}\in\mathbb{Q}$ based on the underlying distribution.
    \item Solve the following \textit{reduced-size scenario problem}
        \begin{align}\label{eq:scenario des in alg1}
        \widehat{\theta}_{N_\ic}= & \arg \underset{\theta\in\Theta}{\text{min}}\quad  c^T\theta\\ \nonumber
        & \text{subject to}\quad  f(\theta,q_d^{(i)})\leq0,\quad i=1,\ldots,N_\ic.
        \end{align}
    \item \textbf{If} the last iteration is reached $(\ic=\ic_t)$, set $\theta\ped{\sol}=\widehat{\theta}_{N_\ic}$ and \textbf{Exit}.
    \item \textbf{Else}, continue to the next step.
  \end{itemize}
  \item \textsc{Validation}
  \begin{itemize}
    \item Draw
    \begin{equation}\label{eq:sample bound Mk}
    M_\ic\ge\frac{\alpha\ln k+\ln \left(\mathcal{S}_{k_t-1}(\alpha)\right)+\ln\frac{2}{\delta}}{\ln\left(\frac{1}{1-\varepsilon}\right)}
    \end{equation}
    i.i.d. samples $\ms_v=\{q_v^{(1)},\ldots, q_v^{(M_\ic)}\}\in\mathbb{Q}$ based on the underlying distribution, and $\mathcal{S}_{\ic_t-1}(\alpha)=\sum_{j=1}^{\ic_t-1}j^{-\alpha}$, where $\alpha>0$ is a tuning parameter.
    \item \textbf{If} $\mathbb{I}_f(\widehat{\theta}_{N_\ic},q_v^{(i)})=0 \text{ for }i=1,\ldots,M_\ic$;
        set $\theta\ped{\sol}=\widehat{\theta}_{N_\ic}$ and \textbf{Exit}.
        \item \textbf{Else}, goto step (\ref{item:update}).
  \end{itemize}
\end{enumerate}
\end{algorithm}
\begin{theorem}\label{theo:property of algorithm 1}
Let Assumptions \ref{assump:convexity} and \ref{assump:unique} hold. Then with probability at least $1-\delta$ the solution obtained from Algorithm \ref{alg:sequential algorithm 1} satisfies the inequality $\ev(\theta\ped{\sol})\leq\varepsilon$.
\end{theorem}
\begin{proof}
See Appendix \ref{app:prof of thm1}.
\end{proof}

We note that in steps 3 and 4, to preserve the i.i.d. assumptions, the design and validation samples need to be redrawn at each iteration, and sample-reuse techniques are not applicable.
\begin{remark}
It is important to observe that the probability
$\ev(\theta\ped{\sol})\leq\varepsilon$ in the statement of Theorem \ref{theo:property of algorithm 1} is the outcome probability of the algorithm. Hence, this probability is a measure on the whole collection of $\sum_k(N_k+M_k)$ samples that includes both design samples and validation samples. This measure is indeed different than the $N$-fold probability measure of the uncertain parameter $q$ which appears in the scenario approach.
\end{remark}
\begin{remark}
The proof of this result has similarities and differences compared to other results which appeared in the probabilistic design literature, see the survey paper \cite{calafiore_research_2011}. Specifically, Theorem 1 in \cite{calafiore_probabilistic_2007} studies the success of a probabilistic oracle, but it does not consider the validation sample techniques. The general framework of sequential algorithms with probabilistic validation is studied in \cite{alamo_randomized_2015}, see Theorem 5 in particular. The contribution of the present paper is to exploit these methods for convex optimization problems in the context of the scenario approach.
\end{remark}
\begin{remark}[Optimal Value of $\alpha$]\label{rem:optimal alpha}{
The  sample bound (\ref{eq:sample bound Mk}) has some similarities with the one derived in \cite[Theorem 2]{calafiore_research_2011},  originally proven in \cite{dabbene_randomized_2010},
and also used in \cite{alamo_randomized_2015}. However, since we are using a finite sum\footnote{See in particular the summation  (\ref{eq:exitbad}) in the proof of Theorem \ref{theo:property of algorithm 1}.}, thanks to the finite scenario bound obtained solving (\ref{eq:binomial dis2}), we can use the finite hyperharmonic series
$\mathcal{S}_{\ic_t-1}(\alpha)=\sum_{j=1}^{\ic_t-1}j^{-\alpha}$
(also known as $p$-series)  instead of the Riemann Zeta function $\sum_{j=1}^{\infty}j^{-\alpha}$.
Indeed, the Riemann Zeta function does not converge when $\alpha$ is smaller than one, while in the presented bound (\ref{eq:sample bound Mk}), $\alpha$ may be smaller than one, which improves the overall sample complexity in particular for large values of $\ic_t$. The optimal value of $\alpha$ which minimizes the sample bound (\ref{eq:sample bound Mk}) has been computed using numerical simulations for different values of the termination parameter $\ic_t$. The ``almost'' optimal value of $\alpha$ minimizing (\ref{eq:sample bound Mk}) for a wide range of $\ic_t$ is $\alpha=0.1$. The bound (\ref{eq:sample bound Mk}) (for $\alpha=0.1)$ improves upon the bound (17) in \cite{calafiore_research_2011}, by $5\% \text{ to } 15\%$ depending on the termination parameter $\ic_t$. It also improves upon the bound in \cite{yasuaki_polynomial-time_2007}, which uses finite sum but in a less effective way.}
\end{remark}
Finally, we note that the dependence of $M_{k}$ upon the parameters $\varepsilon$ and $\delta$ is logarithmic in $1/\delta$
and substantially linear in $1/\varepsilon$. This is a key difference with an approach based on a straightforward (a posteriori) Monte Carlo analysis, which indeed requires $1/\epsilon^2$ validation samples, see e.g.\ \cite{tempo_randomized_2012}.

\subsection{Partial Constraint Satisfaction}\label{sec:partial constraint}
In the ``design" and ``validation" steps of Algorithm \ref{alg:sequential algorithm 1}, \emph{all} elements of the design and validation sample sets are required to satisfy the constraint in (\ref{eq:original problem}). However, it is sometimes impossible to find a solution satisfying the constraint in (\ref{eq:original problem}) for the entire set of uncertainty. For this reason, in Algorithm \ref{alg:sequential algorithm}, we consider the scenario design with discarded constraints where we allow a limited number of design and validation samples to violate the constraint in (\ref{eq:original problem}). We now provide a theorem stating the theoretical properties of Algorithm \ref{alg:sequential algorithm}.
\begin{theorem}\label{theo:property of algorithm}
Let Assumptions \ref{assump:convexity} and \ref{assump:unique} hold. Then with probability at least $1-\delta$ the solution obtained from Algorithm \ref{alg:sequential algorithm} satisfies the inequality $\ev(\theta\ped{\sol})\leq\varepsilon$.
\end{theorem}
\begin{proof}
See Appendix \ref{app:prof of thm2}.
\end{proof}
\begin{algorithm}
\caption{\sc{Sequential Randomized Algorithm: Partial Constraint Satisfaction}}
\label{alg:sequential algorithm}
\begin{enumerate}
  \item \textsc{Initialization} \newline
  Set the iteration counter to zero $(\ic=0)$. Choose  probabilistic levels $\varepsilon$, $\delta$,  number of iterations $\ic_t>1$,  number of discarded constraints $r$ and define the following parameters:
   \begin{align}\label{eq:beta v}
  \beta_v\doteq\max\left\{1,\beta_w\left(\ic_t \ln\frac{2\ic_t}{\delta}\right)^{-1}\right\}, \quad
  \beta_w\doteq\frac{1}{4\varepsilon}\ln\frac{1}{\delta}.
  \end{align}
  \item \textsc{Update}\label{item:update}\newline
  Set $\ic=\ic+1$, $N_\ic\ge N\frac{\ic}{\ic_t}$ and $N_{\ic,r}\ge\frac{(N-r)\ic}{\ic_t}$ where $N$ is the smallest integer satisfying
  \begin{equation}
{r+n_\theta-1 \choose r}\sum_{i=0}^{r+n_\theta-1} { N \choose i}
\varepsilon^i(1-\varepsilon)^{N-i}\leq\delta/2.
\end{equation}
  \item \textsc{Design}
  \begin{itemize}
    \item Draw $N_\ic$ i.i.d. samples $\ms_d=\{q_d^{(1)},\ldots, q_d^{(N_\ic)}\}\in\mathbb{Q}$ based on the underlying distribution.
    \item Solve the following \textit{reduced-size scenario problem}
        \begin{align}\label{eq:scenario des in alg}
        \widehat{\theta}_{N_{\ic,r}}= & \arg \underset{\theta\in\Theta}{\text{min}}\quad  c^T\theta\\ \nonumber
        & \text{subject to}\quad  f(\theta,q_d^{(i)})\leq0,\quad i=1,\ldots,N_{\ic,r}.\footnotemark
        \end{align}
    \item \textbf{If} the last iteration is reached $(\ic=\ic_t)$, set $\theta\ped{\sol}=\widehat{\theta}_{N_\ic,r}$ and \textbf{Exit}.
   \item \textbf{Else}, continue to the next step.
  \end{itemize}
  \item \textsc{Validation}
  \begin{itemize}
    \item Draw
    \begin{equation}
    \label{Mk2}
    M_\ic\ge 2\ic\beta_v\frac{1}{\varepsilon}\ln\frac{2\ic_t}{\delta}
    \end{equation}
    i.i.d. samples
    $\ms_v=\{q_v^{(1)},\ldots, q_v^{(M_\ic)}\}\in\mathbb{Q}$ based on the underlying distribution.
    \item \textbf{If}
        \begin{equation}\label{eq:feasibility condition}
        \frac{1}{M_\ic}\sum _{i=1}^{M_\ic}\mathbb{I}_f(\widehat{\theta}_{N_\ic,r},q_v^{(i)})\leq \left(1-(\ic\beta_v)^{-1/2}\right)\varepsilon
        \end{equation}
        set $\theta\ped{\sol}=\widehat{\theta}_{N_\ic,r}$ and \textbf{Exit}.
        \item \textbf{Else}, goto step (\ref{item:update}).
  \end{itemize}
\end{enumerate}
\end{algorithm}

Algorithm \ref{alg:sequential algorithm} is different from the algorithm presented in \cite{alamo_randomized_2009}, which was derived for non-convex problems, in a number of aspects. That is, the cardinality of the sequence of sample sets used for design and validation increases linearly with iteration counter $\ic$, while it increases exponentially in \cite{alamo_randomized_2009}. Furthermore, the cardinality of the validation sample set at the last iteration $M_{\ic_t}$ in \cite{alamo_randomized_2009} is chosen to be equal to the cardinality of the sample set used for design at the last iteration $N_{\ic_t}$ while, in the presented algorithm $M_{\ic_t}$ and hence $\beta_w$ are chosen based on the additive Chernoff bound which is less conservative.

We also note that both Algorithms \ref{alg:sequential algorithm 1} and \ref{alg:sequential algorithm} fall within the class of SPV algorithms in which the ``design" and ``validation" steps are independent, see  \cite{alamo_randomized_2015}. As a result, in principle we could use the same strategy as Algorithm \ref{alg:sequential algorithm 1} to tackle discarded constraints problems. Nevertheless, Algorithm~\ref{alg:sequential algorithm} appears to be more suitable for discarded constraints problems, since (\ref{eq:scenario des in alg}) forces the solution to violate some constraints.

\subsection{Algorithms Termination and Overall Complexity}

Note that the maximum number of iterations of both Algorithms 1 and 2 is chosen by the user by selecting
the termination parameter $\ic_t$. This choice affects directly the cardinality of the sample sets used for design  $N_\ic$ and validation $M_\ic$ at each iteration, although they converge to fixed values (independent of $\ic_t$) at the last iteration. In problems for which the original scenario sample complexity is large, we suggest to use larger $\ic_t$. In this way, the sequence of sample bounds $N_\ic$ starts from a smaller number and does not increase significantly with the iteration counter $\ic$.
\footnotetext{See footnote 1.}
We also remark that, in  Algorithm \ref{alg:sequential algorithm},  the right hand side of the inequality (\ref{eq:feasibility condition}) cannot be negative, which in turn requires~$\beta_v$ to be greater than one. This condition is taken into account in defining $\beta_v$ in (\ref{eq:beta v}). However, we can avoid generating $\beta_v<1$ by the appropriate choice of $\ic_t$. To this end, we solve the inequality $\beta_v\geq1$ for $\ic_t$ as follows:

\begin{align*}
\beta_v\doteq & \beta_w\left(\ic_t\ln\frac{2\ic_t}{\delta}\right)^{-1}\!\!\!\!\!\!\geq 1\Rightarrow
\ic_t\ln \frac{2\ic_t}{\delta}\leq\beta_w\Rightarrow
\frac{2\ic_t}{\delta}\ln\frac{2\ic_t}{\delta}\leq\frac{2\beta_w}{\delta}.
\end{align*}
\noindent
For implementation purposes, it is useful to use the function ``LambertW" also known as ``Omega function" or ``product logarithm"\footnote{This function is the inverse function of $f(W)=W e^W$. In other words, $W=\text{LambertW}[f(W)]$; see e.g. \cite{corless_lambertw_1996} for more details.}
$
\ic_t\leq \frac{\beta_w}{\text{LambertW}\left(\frac{2\beta_w}{\delta}\right)}.
$
\noindent

Furthermore, note that the overall complexity of Algorithm \ref{alg:sequential algorithm 1} and \ref{alg:sequential algorithm} is a random variable, because the number of iterations is random. Indeed, the number of iterations when the algorithm terminates ($N_\ic$ and $M_\ic$) is only known \emph{a posteriori}, while in the scenario approach we can establish \emph{a priori} sample bounds.
We remark that the computational cost of solving convex optimization problems does not increase linearly with the number of constraints. Hence, we conclude that, if the algorithms terminate with a smaller number of design samples than the original sample complexity of the scenario problem, the reduction in the number of design samples may significantly improve the overall computational cost. This was the case in all the extensive numerical simulations we have performed.

In the particular case when the constraints are linear matrix inequalities (LMIs), then the reduced-size scenario problem (\ref{eq:scenario des in alg1})  can be reformulated as a semidefinite program by combining $N_{k}$ LMIs into a single LMI with block-diagonal structure. It is known, see \cite{BenNem:01}, that the computational cost of this problem with respect to the number of diagonal blocks $N_{k}$ is of the order of $N_{k}^{3/2}$.
Similar discussions hold for Algorithm  \ref{alg:sequential algorithm}.
We conclude that a decrease in $N_{k}$ can significantly reduce the computational complexity.

Finally, note that the computational cost of  validation steps in both presented algorithms is not significant, since they just require \emph{analysis} of a candidate solution for a number of i.i.d. samples extracted from the uncertainty set. For instance, consider the case when $\mathcal{H}_\infty$ performance of an $n$-dimensional system is of concern. This is generally expressed in terms of an LMI arising from a Riccati inequality. In this case, the number of floating point operations required to solve this LMI  inequality is of the order of $n^6$. On the other hand, checking if a Riccati inequality is satisfied requires checking positive definiteness of a symmetric matrix, which is of complexity $n^3$, see
further discussions in \cite[page 1327]{PetTem:14}.

\section{Application to Hard Disk Drive Servo Design}\label{sec:hdd example}
In this section, we employ the developed algorithms to solve a non-trivial application. The problem under consideration is the design of a robust track following controller for a hard disk drive (HDD) servo system affected by parametric uncertainty. Servo system in HDD plays a crucial role in increasing the storage capacity by providing a more accurate positioning algorithm.
The goal in this application is to achieve the storage density of $10$ Tera bit per square inch $(10 Tb/in^2)$.
It requires the variance of the deviation of read/write head from the center of a data track to be less than $1.16$ nanometer. Such a high performance has to be achieved in a robust manner, that is, for all drives produced in a mass production line. On the other hand, some imperfections in the production line such as manufacturing tolerances and slightly different materials or environmental conditions lead to slightly different dynamics over a batch of products.

\begin{table*}[!t]
\begin{center}
\scalebox{0.7}{
\begin{tabular}{cccccccccccccccccc}
\toprule
$\varepsilon$ & $\delta$ & $\ic_t$ & \multicolumn{3}{c}{\emph{Design}} & \multicolumn{3}{c}{\emph{Validation}} & \multicolumn{3}{c}{\emph{Objective}} & \multicolumn{3}{c}{\emph{Iteration}} & \multicolumn{3}{c}{\emph{Computational}}\tabularnewline
  &  &   & \multicolumn{3}{c}{\emph{Samples}} & \multicolumn{3}{c}{\emph{Samples}} & \multicolumn{3}{c}{\emph{Value}} & \multicolumn{3}{c}{\emph{Number}} & \multicolumn{3}{c}{\emph{Time (sec)}}\tabularnewline
\midrule
\midrule
  &  &  &  \emph{Mean} & \emph{Standard} & \emph{Worst} & \emph{Mean} & \emph{Standard} & \emph{Worst} & \emph{Mean} & \emph{Standard} & \emph{Worst} & \emph{Mean} & \emph{Standard} & \emph{Worst} & \emph{Mean} & \emph{Standard} & \emph{Worst}\tabularnewline
  &  &  & & \emph{Deviation} & \emph{Case} &  & \emph{Deviation} & \emph{Case} &  & \emph{Deviation} & \emph{Case}&  & \emph{Deviation} & \emph{Case} &  & \emph{Deviation} & \emph{Case} \tabularnewline
\midrule
$0.2$ & $10^{-2}$ & $20$ &  $219.4$ & $93$ & $496$ & $37$ & $0$ & $37$ & $0.6106$ & $0.006$ & $0.6241$ & $3.54$ & $1.5$& $8$ & $271.8$ & $230.5$ & $1195$\tabularnewline
\midrule
$0.1$ & $10^{-4}$ & $20$ &  $561.3$ & $229.8$ & $1397$ & $121.9$ & $0.37$ & $123$ & $0.6178$ & $0.005$ & $0.6275$ & $4.42$ & $1.81$ & $11$  & $1019$ & $874$ & $6025$\tabularnewline
\midrule
$0.05$ & $10^{-6}$ & $30$ &  $1041$ & $387.8$ & $1747$ & $347.5$ & $0.96$ & $349$ & $0.6211$ & $0.04$ & $0.6281$ & $5.96$ & $2.21$ & $10$  & $2633$ & $1809$ & $6963$\tabularnewline
\bottomrule
\end{tabular}
}
\end{center}
\caption{Simulation Results Obtained Using Algorithm \ref{alg:sequential algorithm 1}}
\label{tab: simulation results}
\end{table*}

\begin{table*}[!t]
\begin{center}
\scalebox{0.7}{
\begin{tabular}{cccccccccccccccccc}
\toprule
$\varepsilon$ & $\delta$ & $\ic_t$ & \multicolumn{3}{c}{\emph{Design}} & \multicolumn{3}{c}{\emph{Validation}} & \multicolumn{3}{c}{\emph{Objective}} & \multicolumn{3}{c}{\emph{Iteration}} & \multicolumn{3}{c}{\emph{Computational}}\tabularnewline
  &  &   & \multicolumn{3}{c}{\emph{Samples}} & \multicolumn{3}{c}{\emph{Samples}} & \multicolumn{3}{c}{\emph{Value}} & \multicolumn{3}{c}{\emph{Number}} & \multicolumn{3}{c}{\emph{Time (sec)}}\tabularnewline
\midrule
\midrule
  &  &  &  \emph{Mean} & \emph{Standard} & \emph{Worst} & \emph{Mean} & \emph{Standard} & \emph{Worst} & \emph{Mean} & \emph{Standard} & \emph{Worst} & \emph{Mean} & \emph{Standard} & \emph{Worst} & \emph{Mean} & \emph{Standard} & \emph{Worst}\tabularnewline
  &  &  & & \emph{Deviation} & \emph{Case} &  & \emph{Deviation} & \emph{Case} &  & \emph{Deviation} & \emph{Case}&  & \emph{Deviation} & \emph{Case} &  & \emph{Deviation} & \emph{Case} \tabularnewline
\midrule
$0.2$ & $10^{-2}$ & $20$ & $141.3$ & $27.9$ & $186$ & $189.2$ & $37.4$ & $249$ & $0.6084$ & $0.005$ & $0.6217$ & $2.28$ & $0.45$ & $3$  & $109.36$ & $41.48$ & $179.53$\tabularnewline
\midrule
$0.1$ & $10^{-4}$ & $20$ &  $276.8$ & $49$ & $381$ & $562$ & $99.6$ & $774$ & $0.6125$ & $0.04$ & $0.6226$ & $2.18$ & $0.36$ & $3$  & $253.3$ & $90.8$ & $456.3$\tabularnewline
\midrule
$0.05$ & $10^{-6}$ & $30$ &  $443.9$ & $93.9$ & $699$ & $1820.2$ & $386.9$ & $2866$ & $0.6169$ & $0.04$ & $0.6253$ & $2.54$ & $0.53$ & $4$  & $600$ & $233.6$ & $1419$\tabularnewline
\bottomrule
\end{tabular}
}
\end{center}
\caption{Simulation Results Obtained Using Algorithm \ref{alg:sequential algorithm}}
\label{tab: simulation results2}
\end{table*}


A voice coil motor (VCM) actuator in a disk drive system can be modeled in the form
\begin{equation}\label{eq:VCM model}
P_{VCM}=\sum_{i=1}^3\frac{A_i}{s^2+2\zeta_i\omega_is+\omega_i^2}
\end{equation}
where $\zeta_i$, $\omega_i$ and $A_i$ are damping ratio, natural frequency and modal constant for each resonance mode, see \cite{chamanbaz_probabilistic_2012} for their nominal values. We assume each natural frequency, damping ratio and modal constant to vary by $5\%$, $5\%$ and $10\%$ from their nominal values respectively. Hence, there are nine uncertain parameters in the plant. 
The objective is to design a full order dynamic output feedback controller which minimizes the worst case $\mathcal{H}_\infty$ norm of the transfer function from disturbance to output. This problem can be  reformulated in terms of linear matrix inequalities. Uncertain parameters enter into the plant description in a non-affine fashion; therefore, classical robust techniques are unable to solve the problem without introducing conservatism.

The sequential algorithms of Section \ref{sec:sequential algorithm} are implemented in Matlab using the toolbox R-RoMulOC \cite{chamanbaz_r-romuloc:_2015}. In the simulations, we assumed the probability density function of all uncertain parameters to be uniform. The choice of uniform distribution is chosen due to its worst case nature \cite{bai_worst-case_1998}. The number of discarded constraints $r$ in Algorithm \ref{alg:sequential algorithm} is chosen to be zero. The resulting optimization problem is solved for different values of $\varepsilon$, $\delta$ and $\ic_t$. Furthermore, we run the simulation $100$ times for each pair. The mean, standard deviation and worst case values of the number of design samples, validation samples, objective value, the iteration number in which the algorithm exits and total computational time\footnote{All the simulations are carried on a work station with $2.83\,GHz$ Core
$2$ Quad CPU and $8\,GB$ RAM.} are tabulated in Table \ref{tab: simulation results} and Table \ref{tab: simulation results2}. We remark that ``design samples" and ``validation samples" in Table \ref{tab: simulation results} and Table \ref{tab: simulation results2} reflect the number of design and validation samples at the iteration when the algorithm exits.  Table \ref{tab: simulation results scenario} shows the scenario bound along with the  computational time required for solving the random convex problem for the same probabilistic levels as Tables \ref{tab: simulation results} and  \ref{tab: simulation results2}; we highlight that the number of design parameters in the problem at hand is $153$.  The average computational time of Tables \ref{tab: simulation results} and \ref{tab: simulation results2} is much smaller than Table \ref{tab: simulation results scenario} which further proves the effectiveness of the proposed sequential randomized algorithms. Nevertheless, there are very rare cases when the computational time of the proposed methodology is larger than that of scenario (last column of Table \ref{tab: simulation results}).
When the probabilistic levels become stringent (last row of Table \ref{tab: simulation results scenario}), we could not solve the scenario problem
while, using the two proposed Algorithms \ref{alg:sequential algorithm 1} and \ref{alg:sequential algorithm} the problem was solved efficiently.
\begin{table}[!t]
\begin{center}
\scalebox{0.9}{
\begin{tabular}{cccc}
\toprule
$\varepsilon$ & $\delta$ & \emph{The Scenario Bound} & \emph{Computational Time (Sec)}\tabularnewline
\midrule
\midrule
$0.2$ & $1\times10^{-2}$ & $1238$ & $538$\tabularnewline
\midrule
$0.1$ & $1\times10^{-4}$ & $2548$ & $1536$\tabularnewline
\midrule
$0.05$ & $1\times10^{-6}$ & $5240$ & $-$\tabularnewline
\bottomrule
\end{tabular}
}
\par\end{center}
\caption{The Scenario Bound and the Required Computational Time for the Same Probabilistic Levels as Tables. \ref{tab: simulation results} and \ref{tab: simulation results2}}
\label{tab: simulation results scenario}
\end{table}
\section{Conclusions}
We proposed two new sequential methods for solving in a computational efficient way uncertain convex optimization problems.
The main philosophy behind the proposed sequential randomized algorithms stems from the consideration that it is easy, from a  computational viewpoint, to validate a given ``candidate solution" for a large number of random samples. The  algorithms have been tested on a numerical example, and extensive numerical simulations show how the total computational effort is ``diluted''  by applying the proposed sequential methodology. Future theoretical work will concentrate on handling unfeasible problems.

\appendices
\section{Proof of the Theorem \ref{theo:property of algorithm 1}}\label{app:prof of thm1}
The proof follows similar reasoning to those in \cite{yasuaki_polynomial-time_2007}.
Notice that Algorithm \ref{alg:sequential algorithm 1}, as constructed, always returns a solution $\theta\ped{\sol}$. Our goal is to bound the probability of this solution being ``bad", that is we want to bound the probability of the event \\
$
\text{ExitBad}\doteq\{\text{Algorithm \ref{alg:sequential algorithm 1} returns a bad solution, i.e.\ $\ev(\theta\ped{sol})>\varepsilon$}\}.
$
To do this, we introduce the following events
\begin{align*}
&\text{Iter}_\ic\doteq\{\text{the validation step of the $\ic$th iteration is reached}\},\\
&\text{Feas}_\ic\doteq\{\widehat{\theta}_{N_\ic}\text{ is declared as feasible in the ``validation" step}\},\\
&\text{Bad}_k\doteq\{\ev(\widehat{\theta}_{N_\ic})>\varepsilon\},\\
&\text{ExitBad}_\ic\doteq\{\text{Algorithm \ref{alg:sequential algorithm 1} exits at iteration }\ic\cap\text{Bad}_\ic\}.
\end{align*}
The goal is to bound the probability of the event ``ExitBad". Since ExitBad$_i\,\cap$ ExitBad$_j=\emptyset$ for $i\neq j$, the probability of the event ``ExitBad" can be reformulated in terms of the event ``ExitBad$_\ic$" as
\begin{align} \nonumber
\Pr\{\text{ExitBad}\}&=\Pr\{\text{ExitBad}_1\cup\text{ExitBad}_2\cup\cdots\cup\text{ExitBad}_{\ic_t}\}\\ \label{eq:exitbad}
&=\Pr\{\text{ExitBad}_1\}+\Pr\{\text{ExitBad}_2\}+\cdots
+\Pr\{\text{ExitBad}_{\ic_t-1}\}+\Pr\{\text{ExitBad}_{\ic_t}\}.
\end{align}
From the definition of the event ``ExitBad$_\ic$" and by considering that to exit at iteration $\ic\leq\ic_{t-1}$, the algorithm needs i) to reach $\ic$th iteration and ii) to declare $\widehat{\theta}_{N_\ic}$ feasible in the validation step, for $\ic=1,\ldots,\ic_{t-1}$, we have
\begin{align*}
&\Pr\{\text{ExitBad}_{\ic}\}=\Pr\{\text{Feas$_\ic\,\cap$ Bad$_\ic\,\cap$ Iter$_\ic$}\}\\
&=\Pr\{\text{Feas$_\ic\,\cap$ Bad$_\ic\,|$ Iter$_\ic$}\}\Pr\{\text{Iter$_\ic$}\}\leq
\Pr\{\text{Feas$_\ic\,\cap$ Bad$_\ic\,|$ Iter$_\ic$}\}\\
&=\Pr\{\text{Feas$_\ic | $ Bad$_\ic\cap$ Iter$_\ic$}\}\Pr\{\text{Bad$_\ic |$ Iter$_\ic$}\}
\leq\Pr\{\text{Feas$_\ic | $ Bad$_\ic\cap$ Iter$_\ic$}\}.
\end{align*}
Using the result of Theorem 1 in \cite{calafiore_probabilistic_2007}, we can bound the right hand side of the last inequality
\begin{equation}\label{eq:temp1}
\Pr\{\text{Feas$_\ic\,| $ Bad$_\ic\,\cap$ Iter$_\ic$}\}<(1-\varepsilon)^{M_\ic}.
\end{equation}
Combining (\ref{eq:exitbad}) and (\ref{eq:temp1}) results in
\begin{align}\nonumber
\Pr\{&\text{ExitBad}\}<(1-\varepsilon)^{M_1}+(1-\varepsilon)^{M_2}+\cdots+
(1-\varepsilon)^{M_{\ic_t-1}}\\
&+\Pr\{\text{ExitBad}_{\ic_t}\}=\sum_{\ic=1}^{\ic_t-1} (1-\varepsilon)^{M_\ic}+\Pr\{\text{ExitBad}_{\ic_t}\}.\label{eq:alg1 sum}
\end{align}
The summation in (\ref{eq:alg1 sum}) can be made arbitrary small by an appropriate choice of $M_\ic$. In particular, by choosing
\begin{equation}\label{eq:p series}
(1-\varepsilon)^{M_\ic}=\frac{1}{\ic^\alpha}\frac{1}{\mathcal{S}_{\ic_t-1}(\alpha)}\frac{\delta}{2},
\end{equation}
we have
\begin{align}
\label{sumMk}
\sum_{\ic=1}^{\ic_t-1} (1-\varepsilon)^{M_\ic}&=\sum_{\ic=1}^{\ic_t-1}\frac{1}{\ic^\alpha}\frac{1}{\mathcal{S}_{\ic_t-1}(\alpha)}\frac{\delta}{2}
=\frac{1}{\mathcal{S}_{\ic_t-1}(\alpha)}\frac{\delta}{2}\sum_{\ic=1}^{\ic_t-1} \frac{1}{\ic^\alpha}
=\frac{\delta}{2}.
\end{align}
Note that the choice of the number of design samples in the last iteration guarantees that $\Pr\{\text{ExitBad}_{\ic_t}\}\le \delta/2$. The statement follows, combining  (\ref{eq:alg1 sum}) with (\ref{sumMk}) and noting that
the bound (\ref{eq:sample bound Mk}) is obtained solving (\ref{eq:p series}) for $M_\ic$.

\section{Proof of the Theorem \ref{theo:property of algorithm}}\label{app:prof of thm2}
To prove the statement, define the events $\text{Iter}_\ic, \text{Feas}_\ic, \text{Bad}_k, \text{ExitBad}_\ic,\, \text{and }
\text{ExitBad}$ as in the proof of Theorem \ref{theo:property of algorithm 1}.
Then, note that the event $\text{Feas}_\ic$ can be written as
\[
\text{Feas}_\ic = \left\{\widehat{V}(\widehat{\theta}_{N_{k,r}},\ms_v)\leq\left(1-(\ic\beta_v)^{-1/2}\right)\varepsilon\right\},
\]
that is, $\widehat{\theta}_{N_{\ic}}$ is declared feasible whenever the feasibility test (\ref{eq:feasibility condition})
is passed.
Again, the goal is to bound the probability of the event ``ExitBad", which can be written as the summation of the events ``ExitBad$_\ic$" as in (\ref{eq:exitbad}).
In turn, for $k\leq k_{t-1}$, we can write
\[
\Pr\{\text{ExitBad}_{\ic}\}=\Pr\{\text{Feas$_\ic\,\cap$ Bad$_\ic\,\cap$ Iter$_\ic$}\}
\leq
\Pr\{\text{Feas$_\ic\,\cap$ Bad$_\ic$}\}
\doteq\Pr\{\text{MisClass}_{\ic}\},
\]
where we denoted  $\text{MisClass}_{\ic}$  the event of  misclassification at iteration $\ic$.
\begin{align*}
\text{MisClass}_{\ic} =&
\left\{
\widehat{V}(\widehat{\theta}_{N_{k,r}},\ms_v)\leq\left(1-(\ic\beta_v)^{-1/2}\right)\varepsilon
\right\}
\cap
\left\{
\ev(\widehat{\theta}_{N_{k,r}})>\varepsilon
\right\},
\quad k=1,\ldots,k_{t-1}.\end{align*}
By defining $\rho_\ic\doteq\left(1-(\ic\beta_v)^{-1/2}\right)\varepsilon$ and $\varepsilon_\ic\doteq(\ic\beta_v)^{-1/2}\varepsilon$, this event can be rewritten as
\begin{align*}
\text{MisClass}_{\ic}& \subseteq
\left\{
\widehat{V}(\widehat{\theta}_{N_{k,r}},\ms_v)\leq \rho_\ic
\right\}
\cap
\left\{
\ev(\widehat{\theta}_{N_{k,r}})-\widehat{V}(\widehat{\theta}_{N_{k,r}},\ms_v)>\varepsilon_\ic
\right\},
\quad k=1,\ldots,k_{t-1}.
\end{align*}
Applying the results of \cite[Theorem 1]{alamo_randomized_2009}, we can bound this event as follows
\begin{align}\label{eq:relative inequality}
\Pr\left\{
\text{MisClass}_{\ic}
\right\}
\leq &\Pr\bigg\{\frac{\ev(\widehat{\theta}_{N_{k,r}})-\widehat{V}(\widehat{\theta}_{N_{k,r}},\ms_v)}
{\sqrt{\ev(\widehat{\theta}_{N_{k,r}})}}>
\frac{\varepsilon_\ic}{\sqrt{\varepsilon_\ic+\rho_\ic}}\bigg\},\quad k=1,\ldots,k_{t-1}.
\end{align}
For any $\eta\in(0,1)$, the one-sided multiplicative Chernoff
inequality \cite{tempo_randomized_2012} guarantees that
\begin{equation}\label{eq:multiplicative chernoff inequality}
\Pr\{\ev(\widehat{\theta}_{N_{k,r}})-\widehat{V}(\widehat{\theta}_{N_{k,r}},\ms_{v})\geq\eta\ev(\widehat{\theta}_{N_{k,r}})\}\leq
e^{\frac{-\ev(\widehat{\theta}_{N_{k,r}})M_{\ic}\eta^2}{2}}.
\end{equation}
\noindent
Setting $\eta=\frac{\varepsilon_\ic}{\sqrt{\varepsilon_\ic+\rho_\ic}}
\frac{1}{\sqrt{\ev(\widehat{\theta}_{N_{\ic,r}})}}$ in (\ref{eq:multiplicative chernoff inequality}), combining with inequality
 (\ref{eq:relative inequality}), we obtain, for $k=1,\ldots,k_{t-1}$:
$
\Pr\left\{
\text{MisClass}_{\ic}
\right\}
\leq e^{\frac{-\varepsilon_\ic^2 M_\ic}{2(\varepsilon_\ic + \rho_\ic)}}
\le \frac{\delta}{2\ic_t},
$
where the last inequality follows from the choice of $M_{k}$ in (\ref{Mk2}).
Notice also that the choice of the number of design samples at the last iteration $N_{k_t}$ guarantees that  the probability of misclassification at the last iteration $(k=k_t)$ is at most $\delta/2$. Therefore, we can write
\[\Pr\{\text{ExitBad}\}\le \sum_{\ic=1}^{\ic_t}
\Pr\{\text{MisClass}_{\ic}\}
\leq\sum_{\ic=1}^{\ic_t-1} \frac{\delta}{2\ic_t}+\Pr\{\text{MisClass}_{\ic_{t}}\} =\frac{\delta(k_t-1)}{2k_t}+\frac{\delta}{2}\leq\delta,
\]
which proves the statement.

\bibliographystyle{plain}
\bibliography{ref,/Users/fabriziodabbene/Dropbox/Biblio/Frugi-biblio}

\end{document}